\def\isarxiv{1} 

\ifdefined\isarxiv
\documentclass[11pt]{article}

\usepackage[numbers]{natbib}

\else
\documentclass[twoside]{article}

\usepackage[accepted]{aistats2025}
\usepackage[round]{natbib}

\fi

\usepackage{amsmath}
\usepackage{amsthm}
\usepackage{amssymb}
\usepackage{algorithm}
\usepackage{subfig}
\usepackage{algpseudocode}
\usepackage{graphicx}
\usepackage{grffile}
\usepackage{wrapfig,epsfig}
\usepackage{url}
\usepackage{xcolor}
\usepackage{epstopdf}

\usepackage{bbm}
\usepackage{dsfont}

\allowdisplaybreaks

\ifdefined\isarxiv

\usepackage{tikz}
\usepackage{hyperref}  
\hypersetup{colorlinks=true,citecolor=blue,linkcolor=blue} 
\usetikzlibrary{arrows}
\usepackage[margin=1in]{geometry}

\else

\usepackage{hyperref}       
\definecolor{mydarkblue}{rgb}{0,0.08,0.45}
\hypersetup{colorlinks=true, citecolor=mydarkblue,linkcolor=mydarkblue}

\fi

\theoremstyle{plain}
\newtheorem{theorem}{Theorem}[section]
\newtheorem{lemma}[theorem]{Lemma}
\newtheorem{definition}[theorem]{Definition}

\newtheorem{assumption}[theorem]{Assumption}

\newcommand{\wh}{\widehat}
\newcommand{\wt}{\widetilde}

\newcommand{\R}{\mathbb{R}}

\renewcommand{\S}{\mathsf{S}}

\renewcommand{\v}{\mathsf{v}}
\newcommand{\m}{\mathsf{m}}
\renewcommand{\d}{\mathsf{d}}
\renewcommand{\H}{\mathsf{H}}

\DeclareMathOperator{\OPT}{OPT}

\DeclareMathOperator{\poly}{poly}

\DeclareMathOperator{\nnz}{nnz}

\DeclareMathOperator{\diag}{diag}

\makeatletter
\newcommand*{\RN}[1]{\expandafter\@slowromancap\romannumeral #1@}
\makeatother

\usepackage{lineno}

\begin{document}

\ifdefined\isarxiv

\date{}

\title{When Can We Solve the Weighted Low Rank Approximation Problem in Truly Subquadratic Time?}
\author{ 
Chenyang Li\thanks{\texttt{
lchenyang550@gmail.com}. Fuzhou University.}
\and
Yingyu Liang\thanks{\texttt{
yingyul@hku.hk}. The University of Hong Kong. \texttt{
yliang@cs.wisc.edu}. University of Wisconsin-Madison.} 
\and
Zhenmei Shi\thanks{\texttt{
zhmeishi@cs.wisc.edu}. University of Wisconsin-Madison.}
\and 
Zhao Song\thanks{\texttt{ magic.linuxkde@gmail.com}. The Simons Institute for the Theory of Computing at the UC, Berkeley.}
}

\else

%
\runningtitle{When Can We Solve the Weighted Low Rank Approximation Problem in Truly Subquadratic Time?}

%

\twocolumn[

\aistatstitle{When Can We Solve the Weighted Low Rank Approximation Problem in Truly Subquadratic Time?}

\aistatsauthor{ 
Chenyang Li$^{1}$
\And 
Yingyu Liang$^{2,3}$
\And  
Zhenmei Shi$^{3}$
\And 
Zhao Song$^{4}$
}

\aistatsaddress{ 
$^1$Fuzhou University. \qquad
$^2$The University of Hong Kong. \qquad
$^3$University of Wisconsin-Madison. 
\qquad
\\
$^4$The Simons Institute for the Theory of Computing at the University of California, Berkeley. 
} 
]


\fi

\ifdefined\isarxiv
\begin{titlepage}
  \maketitle
  \begin{abstract}
The weighted low-rank approximation problem is a fundamental numerical linear algebra problem and has many applications in machine learning. Given a $n \times n$ weight matrix $W$ and a $n \times n$ matrix $A$, the goal is to find two low-rank matrices $U, V \in \mathbb{R}^{n \times k}$ such that the cost of $\| W \circ (U V^\top - A) \|_F^2$ is minimized. Previous work has to pay $\Omega(n^2)$ time when matrices $A$ and $W$ are dense, e.g., having $\Omega(n^2)$ non-zero entries. In this work, we show that there is a certain regime, even if $A$ and $W$ are dense,  we can still hope to solve the weighted low-rank approximation problem in almost linear $n^{1+o(1)}$ time.

  \end{abstract}
  \thispagestyle{empty}
\end{titlepage}

{\hypersetup{linkcolor=black}
}
\newpage

\else

\begin{abstract}
The weighted low-rank approximation problem is a fundamental numerical linear algebra problem and has many applications in machine learning. Given a $n \times n$ weight matrix $W$ and a $n \times n$ matrix $A$, the goal is to find two low-rank matrices $U, V \in \mathbb{R}^{n \times k}$ such that the cost of $\| W \circ (U V^\top - A) \|_F^2$ is minimized. Previous work has to pay $\Omega(n^2)$ time when matrices $A$ and $W$ are dense, e.g., having $\Omega(n^2)$ non-zero entries. In this work, we show that there is a certain regime, even if $A$ and $W$ are dense,  we can still hope to solve the weighted low-rank approximation problem in almost linear $n^{1+o(1)}$ time.
\end{abstract}

\fi

\section{INTRODUCTION}\label{sec:intro}
Weighted low-rank approximation problem~\citep{sj03} is a fundamental numerical linear algebra problem related to matrix completion~\citep{llr16}, faster SVD decomposition~\citep{bjs14,bks16}, and principal component analysis~\citep{xxf+21}. 
It is a prevalent focus over recent years and has been widely applied in machine learning in broad practical applications, such as vision detection~\citep{cxt22}, representation learning~\citep{wwz+19}, image classification~\citep{pcp+23}, language model~\citep{hhc+21}, weather prediction~\citep{wtl18}, and many more.
We can formulate the weighted low-rank approximation problem as below.
\begin{definition}
Given two $n \times n$ matrices $A$ and $W$, the goal is to find two $n \times k$ matrices $U,V$ such that $\| W \circ (UV^\top - A) \|_F^2 \leq (1+\epsilon) \min_{\{A' | \mathrm{rank}(A')=k \}} ~ \| W \circ (A' - A ) \|_F^2$. Here, $W\circ B$ denotes a matrix with entries given by $W_{i,j} B_{i,j}$, where $\circ$ is the Hadamard product operation.
\end{definition}

Here, we provide an example to show that attention computation can be formulated as a weighted low-rank approximation problem. 
Let $Q, K, V \in \mathbb{R}^{n \times d}$ be the query, key, and value matrix. Let $A := \mathsf{Activation}(Q K^\top) \in \mathbb{R}^{n \times n}$ be the attention matrix, e.g, using $\mathsf{Softmax}$ as activation. Let $W \in \mathbb{R}^{n \times n}$ be the attention mask. Then, we have the attention output as $(A \circ W) V$. We can clearly see that $A \circ W$ is the target of the weighted low-rank approximation problem.

There is some well-known hardness in the weighted low-rank approximation problem. \cite{gg11,rsw16} showed that, when we want to find an approximate solution with $\epsilon = 1/\poly(n)$, the weighted low-rank approximation problem is ``equivalent '' to $\mathsf{NP}$-$\mathsf{hard}$ problem.   
To practically fix it, under certain assumptions (e.g., the weighted matrix $W$ has $r$ distinct columns/rows for small $r$), we may have a polynomial time complexity algorithm to solve the weighted low-rank approximation problem.
Thus, there are many theoretical and empirical works studying how we are able to solve the weighted low-rank approximation problem efficiently (e.g., \cite{ev12,zqz+19,syyz23,wy24} and many more).  

However, the best previous works still need to solve weighted low-rank approximation in $\Omega(n^2)$ if $A$ has $\Omega(n^2)$ non-zero entries and $W$ has $\Omega(n^2)$ non-zero-entries. 
Therefore, in this work, we raise a natural fundamental mathematical question:
\begin{center}
{\it
Is there some dense regime setting for $A$ and $W$ so that we can solve the weighted low-rank approximation problem in almost linear $n$ time?
}
\end{center}

The answer is positive. 
When we have mild practical assumptions (see detail in Assumption~\ref{ass:W}), we can solve the weighted low-rank approximation problem in almost linear complexity, i.e., $n^{1+o(1)}$. We state our main results in the following, whose proof is in Section~\ref{sec:recover_solution}.

\begin{theorem}[Main result] \label{thm:main}
Let  $A$ and $ W$ denote two $n \times n$ matrices. 
Assume each entry of $A$ and $W$ needs $n^{\gamma}$ bits\footnote{Each entry in a matrix can be represented by $n^\gamma$ bits. In a real dataset, if we use the float32 format, then as long as $n^\gamma \ge 32$, our assumption holds.} to represent, where $\gamma = o(1)$.
Let $r$ be the number of distinct columns and rows of $W$ (Definition~\ref{def:distinct}).  
Assume $A \circ W$ has at most $r \cdot p$ distinct columns and at most $r \cdot p$ distinct rows, where $p = n^{o(1)}$. 
Assume that $k^2r= O(\log n / \log \log n)$. For every constant $\epsilon>0$, there is an algorithm running in time,   
$n^{1+o(1)}$ time which outputs a factorization (into an $n \times k$ and a $k \times n$ matrix) of a rank-$k$ matrix $\wh{A}$ for which  
\begin{align*}
\|W \circ(A-\wh{A})\|_{F}^{2} \leq(1+\epsilon) \OPT,
\end{align*}
with probability at least $0.99$, where 
\begin{align*}
    \OPT:= \min_{\{A_k| \mathrm{rank}(A_k)=k \}} \| W \circ ( A - A_k ) \|_F^2.
\end{align*}
\end{theorem}

Theorem~\ref{thm:main} shows that, with high probability, we can solve the weighted low-rank approximation problem in almost linear complexity when there is a small number of distinct rows and columns in the target matrix. Furthermore,  for any $\delta \in (0,0.1)$, if we conduct $O(\log(1/\delta))$ repetitions and take the median of results as the final output, our success probability can be boosted to $1-\delta$. Note that it is a standard way to boost the success probability from constant 0.99 to $1-\delta$ for any $\delta \in (0,0.1)$. The high-level idea is a majority vote, which is equivalent to taking the median.

\paragraph{Comparison with Previous Works. }
Our results are beyond \cite{rsw16} in the following sense. Under the condition in Theorem~\ref{thm:main}, their algorithm requires $(\nnz(A) + \nnz(W )) \cdot n^{o(1)} +
n^{1+o(1)}$ running time complexity, where $\nnz(A)$ means number of non-zero entries in $A$.  When $A$ or $W$ are dense matrix, their complexity is $n^{2+o(1)}$, while our complexity is $n^{1+o(1)}$. 
We use a similar high-level proof framework as \cite{rsw16}. However, as our problem setup changes, all analyses need to be updated accordingly. In particular, Theorem~\ref{thm:few_distinct_columns} in Section~\ref{sec:few_distinct_columns} is brand new.
The high-level intuition is that we need to carefully handle the distinct patterns in the analysis to avoid the quadratic time complexity. Then, during the matrix product, we can rewrite the summation order. Then, by fine-grained analysis, we can group many operations together due to the few distinct patterns and bypass the quadratic complexity.

{\bf Roadmap.} 
In Section~\ref{sec:related}, we provide related work of weighted low-rank. In Section~\ref{sec:prelim}, we provide the preliminary of our problem setup and some tools from previous works. Then, in Section~\ref{sec:additive}, we provide a warm up of our analysis to introduce our technique. 
In Section~\ref{sec:relative}, we provide the lower bound of cost used in binary search. 
In Section~\ref{sec:few_distinct_columns}, we introduce how to estimate the $\OPT$ value. 
In Section~\ref{sec:recover_solution}, we can successfully recover our solutions. 
Finally, we state our conclusion in Section~\ref{sec:conclusion}. 
\section{RELATED WORK}\label{sec:related}

\subsection{Weighted Low Rank Applications}
Weighted low-rank techniques have been widely applied in broad machine learning applications such as grayscale-thermal detection~\cite{lwz+16}, object detection~\cite{twzl16,cxt22}, fault detection~\cite{dcz+17}, defect detection~\cite{mww+20,jld+20}, background estimation~\cite{dl17,dlr18}, multi-task learning~\cite{flct18}, robust visual representation learning~\cite{kjn15,wzx+18,wwz+19}, adversarial learning~\cite{lbb+19}, image restoration~\cite{psdx14,cyz+20}, image clustering and classification~\cite{sjs+11,wl17,wllz18,fzcw21,fzc+22,fzc+22b,kkk+23,pcp+23}, robust principal component analysis~\cite{xxf+21}, language models training~\cite{hhw+22}, language model compression~\cite{hhc+21}, weather prediction~\cite{wtl18}, tensor training~\cite{cwc+21,zwht+22}, domain generalization~\cite{smf+24} and many more.
Weighted low-rank techniques have also been widely used in signal processing for filter design and noise removal~\cite{lpw97,lhzc10,jycl15}. 

\subsection{Weighted Low Rank in Attention Mechanism} Particularly, a line of works shows that the attention matrix may have some low-rank property, even under softmax activation function by polynomial approximation methods, e.g., \cite{as23,kll+25_var,lls+25_prune,chl+24_rope,llss24_sparse,lls+24c,lss+24,lssy24,lssz24_tat,swxl24,xsl24,as24_iclr,as24_rope,as24b,hsk+24,hwg+24,hwl+24,hwsl24}. Thus, under such conditions, we may use weighted low-rank approximations for transformers' attention acceleration. For a few distinct columns or rows, empirically, we see that the attention matrix has some good patterns \cite{jlz+24,cls+25,lls+25_grok,cll+25_icl,smn+24}. In this work, we focus on the theoretical analysis and leave its empirical justification as future work.

\subsection{Weighted Low Rank Approximation and Acceleration}
Many previous works try to solve in an efficient way empirically and theoretically~\cite{mmh03,slv04,wj06,mv07,m08,ev12,mxzz13,mu14,rsw16,llr16,d16,dl17,dl17b,bwz19,hlx+19,zqz+19,swz+20,th21,yzls22,syyz23,zyls24}. Particularly, recently, \cite{wy24} proposes an algorithm to output a slightly higher rank output as a proxy to solve the weighted low-rank approximation problem efficiently.
\cite{llr16} develops an efficient framework for alternating minimization to get the weighted low-rank approximation. Similarly, \cite{syyz23} proposes a more robust framework to get the solution.

\subsection{Sketching for Numerical Linear Algebra}
In this work, we use the sketching technique to accelerate a submodular optimization problem. We provide a brief overview of prior sketching work across various domains. Sketching has been utilized to enhance numerous continuous optimization challenges, including linear programming \cite{cls19,song19,b20,jswz21,sy21,gs22}, empirical risk minimization \cite{lsz19,qszz23}, the cutting plane method \cite{jlsw20}, calculating the John Ellipsoid \cite{ccly19,syyz22}, and many more.
Beyond continuous optimization, sketching has also been applied to several discrete optimization issues \cite{dsw22,sxz22,z22,jlsz23}. Additionally, sketching concepts have been implemented in addressing theoretical problems in large language models, such as exponential and softmax regression \cite{lsz23,gsy23,dls23_softmax,llss24}, and reducing the feature dimension of attention matrices \cite{dms23}. Sketching techniques prove valuable in various machine learning tasks, including matrix completion \cite{gsyz23}, adversarial training \cite{gqsw22}, training over-parameterized neural tangent kernel regression \cite{bpsw21,szz21,z22,als+22,hswz22}, matrix sensing \cite{dls23_sensing,qsz23}, kernel density estimation \cite{qrs+22}, and federated learning \cite{swyz23}. Moreover, the application of sketching extends to the theory of relational databases \cite{qjs+22}.
\section{PRELIMINARY}\label{sec:prelim}

We define $[n]:= \{1,2, \dots, n\}$.
Let $\R$ denote the real numbers, and $\R_{\geq 0}$ denote the nonnegative real numbers. Let $\|A\|$   denote the spectral norm of matrix $A$. Let $\|A\|_{F}^{2}=\sum_{i, j} A_{i, j}^{2}$ denote the Frobenius norm of $A$. Let $W \circ A$ denote the entry-wise product of matrices $W$ and $A$. Let $\|A\|_{W}^{2}=\sum_{i, j} W_{i, j}^{2} A_{i, j}^{2}$ denote the weighted Frobenius norm of $A$. Let $\nnz(A)$ denote the number of nonzero entries of $A$. Let $\det(A)$ denote the determinant of a square matrix $A$. Let $A^{\top}$ denote the transpose of $A$. Let $A^{\dagger}$ denote the Moore-Penrose pseudo-inverse of $A$. Let $A^{-1}$ denote the inverse of a full rank square matrix $A$.

For the weight matrix $W$, we always use $W_{*,j}$ to denote the $j$-th column of $W$, and $W_{i,*}$ to denote the $i$-th row of $W$. Let $\diag(W_{*,j})$ denote the diagonal matrix with entries from the vector $W_{*,j}$. Let $\diag(W_{i,*} )$ denote the diagonal matrix with entries from the  vector $W_{i,*}$.

\begin{definition}[Distinct rows and columns]\label{def:distinct}
Given a matrix $R \in m \times n$, we have $R_{*,1}, R_{*,2}, \dots, R_{*,n} \in \R^{m}$ as its column vectors, and $R_{1,*}^\top, R_{2,*}^\top, \dots, R_{m,*}^\top \in \R^{n}$ as its row vectors. We define $r:= |\{ R_{*,1}, R_{*,2}, \dots, R_{*,n} \}|$ and $p := |\{ R_{1,*}^\top, R_{2,*}^\top, \dots, R_{m,*}^\top \}|$, where $|\cdot|$ means the cardinality of a set, i.e., the number of elements in a set. Then, we say $R$ has  $r$ distinct columns and $p$ distinct rows.
\end{definition}

\begin{assumption}\label{ass:W}
Let $r, p \in \mathbb{N}_+$.  We assume $n \times n$ matrix $W$ has $r$ distinct columns and rows. We also assume that $W\circ A$ has $rp$ distinct columns and rows.    
\end{assumption}

In a real application, when the data has strong periodicity, e.g., signal processing, or uniformly draws from a small set, e.g., classification, or the data has a certain structure, e.g., block attention mask in Transformers \citep{vsp+17,flexattn}, the matrix may have a small number of distinct columns or rows.

Let $\mathbb{Z}[x_1, x_2, \cdots,x_v]$ denote the set of multi-variable polynomial with $v$ variables and the coefficients are from $\mathbb{Z}$. Furthermore, if all the coefficients of a polynomial are from $\{0,1, \cdots, 2^H\}$, then we say the bitsize of coefficients of this polynomial is $H$.

\begin{lemma}[Cramer's rule]
Let $R$ be an $n \times n$ invertible matrix. Then, for each $i \in [n]$, $j \in [n]$,
\begin{align*}
    (R^{-1})_{i,j} = \det( R_{ -j }^{-i} ) / \det(R),
\end{align*}
where $R_{ -j }^{-i}$ is the matrix $R$ with the $i$-th row and $j$-th column removed.
\end{lemma}

\subsection{Polynomial System Decision Solver}

Here, we formally state the theorem of the decision problem. 
For a full discussion of
algorithms in real algebraic geometry, we refer the reader to \cite{bpr5}
and \cite{b14}.
\begin{theorem}[Decision Problem \citep{r92a,r92b,bpr96}]\label{thm:decision_problem}
Given a real polynomial system $P (x_{1}, x_{2}, \cdots, x_{\v} )$ having $\v$ variables and $\m$ polynomial constraints $f_{i} (x_{1}, x_{2}, \cdots, x_{\v} ) \Delta_{i} 0, \forall i \in[\m]$, where $\Delta_{i}$ is any of the ``standard relations'': $\{>, \geq,=, \neq, \leq ,<\}$, let $\d$ denote the maximum degree of all the polynomial constraints and let $\H$ denote the maximum bitsize of the coefficients of all the polynomial constraints.  
Then in
\begin{align*}
(\m \d)^{O(\v)} \poly(\H)
\end{align*}
time, one can determine whether a real solution exists for the polynomial system $P$.
\end{theorem}

\subsection{Lower Bound on Cost}
The key result we used for proving the lower bound is the
following bound on the minimum value attained by an integer
polynomial restricted to a compact connected component
of a basic closed semi-algebraic subset of $\R^{\v}$ (see details in Definition~\ref{def:real_field}) defined by polynomials with integer coefficients in terms of the degrees
and the bitsizes of the coefficients of the polynomials involved. 

\begin{theorem}[\cite{jpt13}] \label{thm:jpt13}
Let $T := \{x \in \R^{\v} \mid f_{1}(x) \geq 0, \cdots ,f_{\ell}(x) \geq 0, f_{\ell+1}(x)=0, \cdots, f_{\m}(x)=0 \}$ be defined by polynomials $f_{1}, \cdots, f_{\m} \in \mathbb{Z} [x_{1}, \cdots, x_{\v} ]$ 
with $\v \geq 2$, degrees bounded by an even integer $\d$ and coefficients of absolute value at most $\H$ bitsize, and let $C$ be a compact connected component of $T$. Let $g \in \mathbb{Z} [x_{1}, \cdots, x_{\v} ]$ (here $g$ can be viewed as an objective function) be a polynomial of degree at most $\d$ and coefficients of absolute value bounded by $\H$. Let $\wt{\H}:=\max \{\H, 2 \v+2 \m\}$. Then, the minimum value that $g$ takes over $C$ satisfies that if it is not zero, then its absolute value is $\geq $  
\begin{align*}
 2^{- 2^{3\v \log(\d)} \log( \wt{\H} )}.
\end{align*}
\end{theorem}

\subsection{Upper Bound on Cost}\label{sec:upper_bound}
Now, we provide the upper bound on the $\OPT$, which is defined below. 

\begin{definition}\label{def:opt}
Given $A, W \in \R^{n \times n}$ and $k \in [n]$, we define $\OPT$ as
    \begin{align*}
        \OPT := \min _{U \in \R^{n \times k}, V \in \R^{n \times k}}\|(U V^\top-A) \circ W\|_F^2 .
    \end{align*}
\end{definition}

\begin{lemma}[folklore]\label{lem:opt}
Let $A, W, \OPT$ be defined in Definition~\ref{def:opt}. Let each entry of $A,W$ can be represented by $n^{\gamma}$ bits, with $\gamma \in (0,1)$. Then, we have
\begin{align*}
\OPT \leq \poly(n) \cdot 2^{n^{\gamma}}.
\end{align*}
\end{lemma}

\begin{proof}
Set $U$ and $V$ to be the zero matrices.
\end{proof}

\subsection{Sketching Tool}
 
We use a sketching tool from previous work.
\begin{lemma}[Theorem 3.1 in \cite{rsw16}]
Let $A^1, \ldots, A^m \in \R^{n \times k}$ be $m$ matrices of size $n \times k$. Let $b^1, \ldots, b^m \in \R^{n \times 1}$ be $m$ column vectors of dimension $n$. For $1 \leq i \leq m$ denote: $x^i= \arg\min_{ x \in \R^{k} } \|A^i x-b^i \|_2^2$ the solution of the $i$-th regression problem. Let $S \in \R^{t \times n}$ be a random matrix with i.i.d. Gaussian entries with zero mean and standard deviation $1 / \sqrt{t}$. For $1 \leq i \leq m$ denote $ y^i= \arg\min_{y \in \R^{k}} \|S A^i y-S b^i \|_2^2 $  the solution of the $i$-th regression problem in the sketch space. We claim that for every $0<\epsilon<1 / 2$, with high probability, one can set $t=O(k / \epsilon)$ such that: 
\begin{align*}
\sum_{i=1}^m \|A^i y^i-b^i \|_2^2 \leq(1+\epsilon) \cdot \sum_{i=1}^m \|A^i x^i-b^i \|_2^2.
\end{align*}
\end{lemma}

\subsection{Backgrounds on Semi-Algebraic Sets}\label{sec:preli_app}

The following real algebraic geometry definitions are needed when proving a lower bound for the minimum nonzero cost of our problem. For a full discussion, we refer the reader to \cite{bcr87}. Here, we use the brief summary by \cite{bpr5}.

\begin{definition}[\cite{bpr5}]\label{def:real_field}
Let $R$ be a real closed field. 
Given $x= (x_{1}, \cdots, x_{v} ) \in R^{v}, r \in R, r>0$, we denote
\begin{align*}
B_{v}(x, r) := & ~ \{y \in R^{v} |\| y-x \|^{2}<r^{2} \}  \mathrm{~~~ (open~ball), } \\
\bar{B}_{v}(x, r) := & ~ \{y \in R^{v}| \| y - x \|^{2} \leq r^{2} \}  \mathrm{~~~ (closed~ball) }.
\end{align*}

$A$ set $S \subset R^{v}$ is open if it is the union of open balls, i.e., if every point of $U$ is contained in an open ball contained in $U$.

A set $S \subset R^{v}$ is closed if its complement is open. Clearly, the arbitrary union of open sets is open, and the arbitrary intersection of closed sets is closed.

Semi-algebraic sets are defined by a finite number of polynomial inequalities and equalities.

A semi-algebraic set has a finite number of connected components, each of which is semi-algebraic. Here, we use the topological definition of a connected component, which is a maximal connected subset (ordered by inclusion), where connected means it cannot be divided into two disjoint nonempty closed sets.

A closed and bounded semi-algebraic set is compact.

A semi-algebraic set $S \subset R^{v}$ is semi-algebraically connected if $S$ is not the disjoint union of two non-empty semialgebraic sets that are both closed in S. Or, equivalently, $S$ does not contain a non-empty semi-algebraic strict subset which is both open and closed in $S$. 

A semi-algebraically connected component of a semi-algebraic set $S$ is a maximal semi-algebraically connected subset of $S$. 
\end{definition}

\section{WARMUP}\label{sec:additive}

In this section, to demonstrate the new technique, we prove the following theorem.

\begin{lemma}\label{lem:warmup}
Given two $n \times n$ size matrices $A$ and $W$, $1 \leq k \leq n$  such that: Let each entry of $A,W$ can be represented by $n^{\gamma}$ bits, with $\gamma \in (0,1)$; Let $\OPT$ be defined as Definition~\ref{def:opt}.  
Then, we can show
\begin{itemize}
\item {\bf Part 1.} There is an algorithm that runs in $2^{O(nk\log n)}$ time, and outputs a number $\Lambda$  such that $\OPT \leq \Lambda \leq 2 \OPT $. 
\item {\bf Part 2.} There is an algorithm that runs in $2^{O(n k \log n)}$ time and returns $U \in \R^{n \times k}$ and $V \in \R^{n \times k}$ such that $\| (UV^\top - A) \circ W \|_F^2 \leq 2 \OPT$.
\end{itemize}
\end{lemma}

\begin{proof}

{\bf Proof of Part 1.}

We can create $2nk$ variables to explicitly represent each entry of $U$ and $V$. Let $g(x) =  \| W \circ ( U V^\top - A ) \|_F^2$. Let $L = 2^{n^{\gamma}}$. Then, we can write down a polynomial system (the decision problem defined in Theorem~\ref{thm:decision_problem})
\begin{align*}
   \min & ~  g(x) \\
    \mathrm{~s.t.~} & ~  U_{i,j} \in [-L,L], \forall i,j \\
    & ~ V_{i,j} \in [-L,L], \forall i,j
\end{align*}
Using Theorem~\ref{thm:jpt13}, we know the above system has
\begin{align*}
    \m = 2nk, \v = 2nk, \d = 4, \H = n^{\gamma}, \wt{\H} = n^{\gamma} + O(nk).
\end{align*}
The lower bound on $g(x)$ (if $g(x)$ is not zero) is going to be  
\begin{align*}
    c_{\mathrm{lower}} = & ~ 2^{-2^{3\v \log (\d)} \log (\wt{H}) } \\
    \geq & ~  ~ 2^{-2^{ O( nk )} \log(n^{\gamma} + nk) } \\
    \geq & ~  ~ 2^{-2^{O(nk \log n)}}.
\end{align*}

By Lemma~\ref{lem:opt}, we know the upper bound is $C_{\mathrm{upper}} = \poly(n) \cdot 2^{n^{\gamma}}$.
After knowing the lower bound and upper bound on cost, the number of binary search iterations is upper-bounded by
\begin{align*}
\log( \frac{C_{\mathrm{upper}}}{ c_{\mathrm{lower}} } ) 
= ~ \log( \frac{ \poly(n) 2^{n^{\gamma}} }{ 2^{-2^{ O(nk \log n)} } } ) 
\leq  ~ 2^{O(nk \log n)}.
\end{align*}

In each of the above iterations, we need to run Theorem~\ref{thm:decision_problem} with the system 
\begin{align*}
    \mathrm{~s.t.~} 
    & ~  g(x) \in [\Gamma_t, 2 \Gamma_t] , ~ U_{i,j} \in [-L, L] , ~ V_{i,j} \in [-L, L]
\end{align*}
with parameters,
\begin{align*}
\m = 2nk+1,  \v = 2nk, \d =4, \H = n^{\gamma}.
\end{align*}
Then, running time complexity is
\begin{align*}
 (\m \d)^{O(\v)} \cdot \poly(\H) 
 =  & ~ (10nk)^{O(nk)} \cdot \poly(n^{\gamma}) \\
 =  & ~ 2^{O( n k \log n) }.
\end{align*}

Thus, combining the number of iterations and time for each iteration, we can find the number $\Gamma \in [\OPT, 2\OPT]$.

{\bf Proof of Part 2.}

Next, similar to {\bf Part 1}, we need to repeat the binary search for $2nk$ times for each variable in $U$ and $V$, and each time, the number of total binary search steps is $n^{\gamma}$. Thus, we can output the $U, V$ in the same running time as finding $\Gamma$.
\end{proof}

In the next few sections, we will explain how to reduce the number of variables and how to reduce the number of constraints.
\section{LOWER BOUND ON OPT}\label{sec:relative}

We assume that $W$ has $r$ distinct rows and $r$ distinct columns. Then, we get rid of the dependence on $n$ in the degree.

\begin{theorem}[Implicitly in \cite{rsw16}]\label{thm:lower_bound_on_cost}
Assuming that $W$ has $r$ distinct rows and $r$ distinct columns, each entry of $A$ and $W$ needs $n^{\gamma}$ bits to represent. Assume $\OPT >0$. Then we know that, with high probability,  
\begin{align*}
    \OPT \geq 2^{-n^{\gamma} 2^{\wt{O}(rk^2/\epsilon)} }.
\end{align*}
\end{theorem}

\begin{proof}

We use $A_{i,*} \in \R^n$ denote the $i$-th row of $A$. We use $W_{i,*} \in \R^n$ to denote the $i$-th row of $W$. Let $(U_1)_{i,*}$ denote the $i$-th row of $U_1$.  
For any $n \times k$ matrix $U_1$ and for any $k \times n$ matrix $Z_1$, we have
\begin{align*}
 & ~ \| (U_1 Z_1 - A) \circ W \|_F^2 \\
 = & ~ \sum_{i=1}^n \| (U_1)_{i,*} Z_1 \diag( W_{i,*} ) - A_{i,*} \diag( W_{i,*} ) \|_2^2.
\end{align*}

Based on the observation that $W$ has $r$ distinct rows, we use group $g_{1,1}, g_{1,2}, \cdots , g_{1,r}$ to denote $r$ disjoint sets such that
\begin{align*}
    \cup_{i=1}^r g_{1,i} = [n]
\end{align*}
For any $i \in [r]$, for any $j_1, j_2 \in g_{1,i}$, we have $W_{j_1,*} = W_{j_2,*}$.

Thus, we can have 
\begin{align*}
    & ~ \sum_{i=1}^n \| (U_1)_{i,*} Z_1 \diag( W_{i,*} ) - A_{i,*} \diag( W_{i,*} ) \|_2^2 \\
    = & ~ \sum_{i=1}^r  \sum_{\ell \in g_{1,i}} \| (U_1)_{\ell,*} Z_1 \diag( W_{\ell,*} ) - A_{\ell,*} \diag( W_{\ell,*} ) \|_2^2  .
\end{align*}
We can sketch the objective function by choosing Gaussian matrices $S_1 \in \R^{n \times s_1}$ with $s_1 = O(k/\epsilon)$.
\begin{align*}
\sum_{i=1}^n \| (U_1)_{i,*} Z_1 \diag( W_{i,*} ) S_1 - A_{i,*} \diag( W_{i,*} ) S_1 \|_2^2 .
\end{align*}
Let $\wh{U}_1$ denote the optimal solution of the sketch problem,

\begin{align*}
    & ~ \wh{U}_1 =  \arg\min_{U_1 \in \R^{n \times k}} \\
    & ~ \sum_{i=1}^n \| (U_1)_{i,*} Z_1 \diag( W_{i,*} ) S_1 - A_{i,*} \diag( W_{i,*} ) S_1 \|_2^2 .
\end{align*}
By properties of $S_1$, plugging $\wh{U}_1$ into the original problem, we obtain
\begin{align*}
& ~ \sum_{i=1}^r  \sum_{\ell \in g_{1,i}} \| (\wh{U}_1)_{\ell,*} Z_1 \diag( W_{\ell,*} ) - A_{\ell,*} \diag( W_{\ell,*} ) \|_2^2  \\
\leq & ~  (1+\epsilon) \cdot \OPT.
\end{align*}

Let $R$ denote the set of all $\S(g_{1,i})$ (for all $i \in [r]$ and $|R| = r$)

Note that $\wh{U}_1$ also has the following form, for each $\ell \in L \subset [n]$ (Note that $|L| = r p$.)
\begin{align*}
(\wh{U}_1)_{\ell, *} 
= & ~ A_{\ell,*} \diag(W_{\ell,*}) S_1 \cdot ( Z_1 \diag( W_{\ell,*} ) S_1 )^\dagger \\
= & ~  A_{\ell,*} \diag(W_{\ell,*}) S_1 \cdot ( Z_1 \diag( W_{\ell,*} ) S_1 )^\top \\
& ~ \cdot ( ( Z_1 \diag( W_{\ell,*} ) S_1 ) ( Z_1 \diag( W_{\ell,*} ) S_1 )^\top )^{-1}.
\end{align*}

Recall the number of different $ \diag(W_{\ell,*})$ is at most $r$.

For each $k \times s_1$ matrix $Z_1 \diag( W_{\ell,*}) S_1$, we create $k \times s_1$ variables to represent it. Thus, we create $r$ matrices,
\begin{align*}
\{ Z_1 D_{W_{i,*}} S_1 \}_{i \in R}.
\end{align*}
For simplicity, let $P_{1,i} \in \R^{k \times s_1}$ denote $Z_1 \diag(W_{i,*}) S_1$. Then we can rewrite $\wh{U}^i$ as follows
\begin{align*}
    \wh{U}_1^i = A_{i,*} \diag(W_{i,*}) S_1  \cdot P_{1,i}^\top (P_{1,i} P_{1,i}^\top )^{-1}.
\end{align*}
If $P_{1,i} P_{1,i}^\top \in \R^{k \times k}$ has rank-$k$, then we can use Cramer's rule to write down the inverse of $P_{1,i} P_{1,i}^\top$.
For the situation, it is not full rank. We can guess the rank. Let $t_i \leq k$ denote the rank of $P_{1,i}$. Then, we need to figure out a maximal linearly independent subset of columns of $P_{1,i}$. We can also guess all the possibilities, which is at most $2^{O(k)}$. Because we have $r$ different $P_{1,i}$, the total number of guesses we have is at most $2^{O(rk)}$. Thus, we can write down $(P_{1,i} P_{1,i}^\top)^{-1}$ according to Cramer's rule. Note that $(P_{1,i} P_{1,i}^\top)^{-1}$ can be view as $P_a/P_b$ where $P_a$ is a polynomial and $P_b$ is another polynomial which is essentially $\det(P_{1,i} P_{1,i}^\top)$.

After $\wh{U}_1$ is obtained, we fix $\wh{U}_1$. We consider 
\begin{align*}
    & ~ \wh{U}_2 =  \arg\min_{U_2 \in \R^{n \times k}}  \| (\wh{U}_1 U_2^\top - A) \circ W \|_F^2 .
\end{align*}
In a similar way, we can get and write $\wh{U}_2$. 

Overall, by creating $l = O(rk^2/\epsilon)$ variables, we have rational polynomials $\wh{U}_1(x)$ and $\wh{U}_2(x)$. Note that $\wh{U}_1(x)$ only has $rp$ different rows, and same for $\wh{U}_2(x)$.

 Indeed, now we have only $2 r$ distinct denominators (w.l.o.g., assume the first $r$ columns are distinct and the first $r$ rows are distinct),
\begin{align*}
h_{1,i}(x) & = \det ( P_{1,i} P_{1,i}^\top ), \forall i \in[r] \\
h_{2,i}(x) & = \det ( P_{2,i} P_{2,i}^\top ), \forall i \in[r].
\end{align*}

Then, we can write down the following optimization problem,
\begin{align*}
\min _{x \in \R^l} & ~ p(x) / q(x) \\
\mathrm{s.t. } & ~ h_{1,i}^2(x) \neq 0, h_{2,i}^2(x) \neq 0, \forall i \in[r], \\
& ~ q(x)=\prod_{i=1}^r h_{1,i}^2(x) h_{2,i}^2(x),
\end{align*}
where $q(x)$ has degree $O(r k)$, the maximum coefficient in absolute value is 
\begin{align*}
( 2^{n^{\gamma}} )^ {O (r k )}
\end{align*}
 and the number of variables $O (r k^2 / \epsilon )$. However, that formulation $p(x)/q(x)$ is not a polynomial, to further make it a polynomial, we introduce variable $y$:
 \begin{align*}
\min _{x \in \R^l} & ~ p(x) y \\
\mathrm{s.t. } & ~ h_{1,i}^2(x) \neq 0, h_{2,i}^2(x) \neq 0, \forall i \in[r], \\
& ~ q(x)=\prod_{i=1}^r h_{1,i}^2(x) h_{2,i}^2(x) \\
& ~ q(x)y - 1 =0.
\end{align*}

Applying Theorem~\ref{thm:jpt13}, with parameters
\begin{align*}
& \m= O(r), \v = O(rk^2/\epsilon),  \d = O(r) , \\
& \H = 2^{O( n^{\gamma} rk ) } ,  \wt{H} = O(\H) .
\end{align*}

we can achieve the following minimum nonzero cost: 
\begin{align*}
\geq  ~ 2^{-2^{3 \v \log (\d)} \log( \wt{\H}) } 
\geq  ~ 
2^{-n^\gamma 2^{\wt{O} (r k^2 / \epsilon )}  } .
\end{align*}
 Thus, we complete the proofs.
\end{proof}

\section{FEW DISTINCT COLUMNS}\label{sec:few_distinct_columns}

In this section, we try to estimate the $\OPT$ value. 

\begin{theorem}\label{thm:few_distinct_columns}
Given a matrix $A$ and $W$, each entry can be written using $O(n^{\gamma})$ bits for $\gamma >0$.
Given $A \in \R^{n \times n}, W \in \R^{n \times n}, 1 \leq k \leq n$.  Assume that $W$ has $r$ distinct columns and rows.
   
Then, with high probability,  
one can output a number $\Lambda$ in time $n^{1+\gamma} \cdot p \cdot 2^{O (k^2 r / \epsilon )} $ such that 
\begin{align*}
    \OPT \leq \Lambda \leq(1+\epsilon) \OPT .
\end{align*}
\end{theorem}
\begin{proof}

Let $U_1^*, U_2^* \in \R^{n \times k}$ denote the matrices satisfying 
\begin{align*}
    \| W \circ (U_1^* (U_2^*)^\top - A) \|_F^2 = \OPT.
\end{align*}

We use $A_{i,*} \in \R^n$ denote the $i$-th row of $A$. We use $W_{i,*} \in \R^n$ to denote the $i$-th row of $W$. Let $(U_1)_{i,*}$ denote the $i$-th row of $U_1$. Let $Z_1 = (U_2^*)^\top$.  
For any $n \times k$ matrix $U_1$,
\begin{align*}
 & ~ \| (U_1 Z_1 - A) \circ W \|_F^2 \\
 = & ~ \sum_{i=1}^n \| (U_1)_{i,*} Z_1 \diag( W_{i,*} ) - A_{i,*} \diag( W_{i,*} ) \|_2^2.
\end{align*}

Based on the observation that $W$ has $r$ distinct rows, we use group $g_{1,1}, g_{1,2}, \cdots , g_{1,r}$ to denote $r$ disjoint sets such that
\begin{align*}
    \cup_{i=1}^r g_{1,i} = [n]
\end{align*}
For any $i \in [r]$, for any $j_1, j_2 \in g_{1,i}$, we have $W_{j_1,*} = W_{j_2,*}$.

Next, based on assumptions on $W$ and $A$, we use $g_{1,i,1}$, $g_{1,i,2}$, $g_{1,i,p}$ to denote $p$ groups such that 
\begin{align*}
    \cup_{j=1}^p g_{1,i,j} = g_{1,i} .
\end{align*}
For any $i \in [r]$, for any $j \in [p]$, for any $\ell_1, \ell_2 \in g_{1,i,j}$, we have  $(W_{\ell_1,*} \circ A_{\ell_1,*}) = (W_{\ell_2,*} \circ A_{\ell_2,*})$.

Let $\S(g_{1,i,j})$ denote the smallest index from set $g_{1,i,j}$

Thus, we can have 
\begin{align*}
    & ~ \sum_{i=1}^n \| (U_1)_{i,*} Z_1 \diag( W_{i,*} ) - A_{i,*} \diag( W_{i,*} ) \|_2^2 \\
    = & ~ \sum_{i=1}^r \sum_{j =1}^p \sum_{\ell \in g_{1,i,j}} \| (U_1)_{\ell,*} Z_1 \diag( W_{\ell,*} ) \\
     & ~ \quad\quad\quad\quad\quad\quad - A_{\ell,*} \diag( W_{\ell,*} ) \|_2^2 \\
    = & ~\sum_{i=1}^r \sum_{j =1}^p |g_{1,i,j}| \cdot \| (U_1)_{\S(g_{1,i,j}),*} Z_1 \diag( W_{ \S(g_{1,i,j}),*} ) \\
    & ~ \quad\quad\quad\quad\quad\quad - A_{\S(g_{1,i,j}),*} \diag( W_{\S(g_{1,i,j}),*} ) \|_2^2 .
\end{align*}
We can sketch the objective function by choosing Gaussian matrices $S_1 \in \R^{n \times s_1}$ with $s_1 = O(k/\epsilon)$.
\begin{align*}
& ~ \sum_{i=1}^r \sum_{j =1}^p |g_{1,i,j}| \cdot \| (U_1)_{\S(g_{1,i,j}),*} Z_1 \diag( W_{ \S(g_{1,i,j}),*} ) S_1 \\
& ~ \quad\quad\quad\quad\quad\quad - A_{\S(g_{1,i,j}),*} \diag( W_{\S(g_{1,i,j}),*} ) S_1 \|_2^2.
\end{align*}
Let $\wh{U}_1$ denote the optimal solution of the sketch problem,
\begin{align*}
    \wh{U}_1 = & ~ \arg\min_{U_1} \sum_{i=1}^r \sum_{j =1}^p |g_{1,i,j}| \\
    & ~ \cdot \| (U_1)_{\S(g_{1,i,j}),*} Z_1 \diag( W_{ \S(g_{1,i,j}),*} ) S_1 \\
    & ~ \quad - A_{\S(g_{1,i,j}),*} \diag( W_{\S(g_{1,i,j}),*} ) S_1 \|_2^2.
\end{align*}
By properties of $S_1$, plugging $\wh{U}_1$ into the original problem, we obtain
\begin{align*}
& ~ \sum_{i=1}^r \sum_{j =1}^p |g_{1,i,j}| \cdot \| (U_1)_{\S(g_{1,i,j}),*} Z_1 \diag( W_{ \S(g_{1,i,j}),*} ) \\
& ~ \quad\quad - A_{\S(g_{1,i,j}),*} \diag( W_{\S(g_{1,i,j}),*} ) \|_2^2 \leq (1+\epsilon) \cdot \OPT.
\end{align*}

Let $R$ denote the set of all $\S(g_{1,i})$ (for all $i \in [r]$ and $|R| = r$).

Let $L$ denote the set of all $\S(g_{1,i,j})$ (for all $i \in [r]$, $j \in [p]$ and $|L| = rp$).

Note that $\wh{U}_1$ also has the following form, for each $\ell \in L \subset [n]$ (Note that $|L| = r p$.)
\begin{align*}
(U_1)_{\ell, *} 
= & ~ A_{\ell,*} \diag(W_{\ell,*}) S_1 \cdot ( Z_1 \diag( W_{\ell,*} ) S_1 )^\dagger \\
= & ~  A_{\ell,*} \diag(W_{\ell,*}) S_1 \cdot ( Z_1 \diag( W_{\ell,*} ) S_1 )^\top \\
& ~ \cdot ( ( Z_1 \diag( W_{\ell,*} ) S_1 ) ( Z_1 \diag( W_{\ell,*} ) S_1 )^\top )^{-1}.
\end{align*}

Recall the number of different $A_{\ell,*} \diag(W_{\ell,*})$ is at most $rp$, and the number of different $ \diag(W_{\ell,*})$ is at most $r$. For each $k \times s_1$ matrix $Z_1 \diag( W_{\ell,*}) S_1$, we create $k \times s_1$ variables to represent it. Thus, we create $r$ matrices,
\begin{align*}
\{ Z_1 \diag( W_{i,*} ) S_1 \}_{i \in R}. 
\end{align*}
For simplicity, let $P_{1,i} \in \R^{k \times s_1}$ denote $Z_1 \diag(W_{i,*}) S_1$. Then we can rewrite $(\wh{U}_1)_{i,*}$ as follows
\begin{align*}
    (\wh{U}_1)_{i,*} = A_{i,*} \diag(W_{i,*}) S_1  \cdot P_{1,i}^\top (P_{1,i} P_{1,i}^\top )^{-1}.
\end{align*}
If $P_{1,i} P_{1,i}^\top \in \R^{k \times k}$ has rank-$k$, then we can use Cramer's rule to write down the inverse of $P_{1,i} P_{1,i}^\top$.
For the situation, it is not full rank. We can guess the rank. Let $t_i \leq k$ denote the rank of $P_{1,i}$. Then, we need to figure out a maximal linearly independent subset of columns of $P_{1,i}$. We can also guess all the possibilities, which is at most $2^{O(k)}$. Because we have $r$ different $P_{1,i}$, the total number of guesses we have is at most $2^{O(rk)}$. Thus, we can write down $(P_{1,i} P_{1,i}^\top)^{-1}$ according to Cramer's rule. Note that $(P_{1,i} P_{1,i}^\top)^{-1}$ can be view as $P_a/P_b$ where $P_a$ is a polynomial and $P_b$ is another polynomial which is essentially $\det(P_{1,i} P_{1,i})$.

After $\wh{U}_1$ is obtained, we will fix $\wh{U}_1$ in the next round.

For any $n \times k$ matrix $U_2$, we can rewrite 
\begin{align*}
 & ~ \| (\wh{U}_1 U_2^\top  - A ) \circ \|_F^2 \\
 = & ~ \sum_{i=1}^n \| \diag(  W_{*,i} ) \wh{U}_1   (U_2^\top)_{*,i}   - \diag(W_{*,i}) A_{*,i} \|_2^2 .
\end{align*}

Based on the observation that $W$ has $r$ columns rows, we use group $g_{2,1}, g_{2,2}, \cdots , g_{2,r}$ to denote $r$ disjoint sets such that
\begin{align*}
    \cup_{i=1}^r g_{2,i} = [n].
\end{align*}
For any $i \in [r]$, for any $j_1, j_2 \in g_{2,i}$, we have $W_{*,j_1} = W_{*,j_2}$.

Next, based on assumptions on $W$ and $A$, we use $g_{2,i,1}$, $g_{2,i,2}$, $g_{2,i,p}$ to denote $p$ groups such that 
\begin{align*}
    \cup_{j=1}^p g_{2,i,j} = g_{2,i}.
\end{align*}
For any $i \in [r]$, for any $j \in [p]$, for any $\ell_1, \ell_2 \in g_{2,i,j}$, we have  $(W_{*,\ell_1} \circ A_{*,\ell_1}) = (W_{*,\ell_2} \circ A_{*,\ell_2})$.

Let $\S(g_{2,i,j})$ denote the smallest index from set $g_{2,i,j}$

Thus, we can have 
\begin{align*}
    & ~ \sum_{i=1}^n \| \diag( W_{*,i} ) \wh{U}_1   (U_2^\top )_{*,i}  -  \diag( W_{*,i} ) A_{*,i} \|_2^2 \\
    = & ~ \sum_{i=1}^r \sum_{j =1}^p \sum_{\ell \in g_{1,i,j}} \|  \diag( W_{*,\ell} ) \wh{U}_1 (U_2^\top )_{*,\ell} \\
    & ~ \quad\quad\quad\quad\quad\quad - \diag( W_{*,\ell} ) A_{*,\ell}  \|_2^2 \\
    = & ~\sum_{i=1}^r \sum_{j =1}^p |g_{2,i,j}| \cdot \|  \diag( W_{ *, \S(g_{2,i,j}) } ) \wh{U}_1 (U_2^\top )_{*, \S(g_{2,i,j})} \\
    & ~ \quad\quad\quad\quad\quad\quad -  \diag( W_{*, \S(g_{1,i,j})} ) A_{*,\S(g_{1,i,j})} \|_2^2 .
\end{align*}
We can sketch the objective function by choosing Gaussian matrices $S_2 \in \R^{n \times s_1}$ with $s_2 = O(k/\epsilon)$.
\begin{align*}
& ~ \sum_{i=1}^r \sum_{j =1}^p |g_{2,i,j}| \cdot \| S_2  \diag( W_{ \S(g_{2,i,j}),*} ) \wh{U}_1  (U_2^\top )_{*, \S(g_{2,i,j}) } \\
& ~ \quad\quad\quad\quad - S_2 \diag( W_{*, \S(g_{2,i,j})} )  A_{*, \S(g_{2,i,j}) }  \|_2^2.
\end{align*}
Let $\wh{U}_1$ denote the optimal solution of the sketch problem,
\begin{align*}
    \wh{U}_1 = & ~ \arg\min_{U_1} \sum_{i=1}^r \sum_{j =1}^p |g_{2,i,j}| \\
    & ~  \cdot \| S_2  \diag( W_{ \S(g_{2,i,j}),*} ) \wh{U}_1  (U_2^\top )_{*, \S(g_{2,i,j}) } \\
    & ~ - S_2 \diag( W_{*, \S(g_{2,i,j})} )  A_{*, \S(g_{2,i,j}) }  \|_2^2.
\end{align*}
By properties of $S_1$, plugging $\wh{U}_2$ into the original problem, we obtain
\begin{align*}
& ~ \sum_{i=1}^r \sum_{j =1}^p |g_{2,i,j}| \cdot \|  \diag( W_{ \S(g_{2,i,j}),*} ) \wh{U}_1  ( \wh{U}_2^\top )_{*, \S(g_{2,i,j}) } \\
& ~ -  \diag( W_{*, \S(g_{2,i,j})} )  A_{*, \S(g_{2,i,j}) }  \|_2^2\leq (1+\epsilon) \cdot \OPT.
\end{align*}

Let $R$ denote the set of all $\S(g_{1,i})$ (for all $i \in [r]$ and $|R| = r$).

Let $L$ denote the set of all $\S(g_{1,i,j})$ (for all $i \in [r]$, $j \in [p]$ and $|L| = rp$).

Note that $\wh{U}_1$ also has the following form, for each $\ell \in L \subset [n]$ (Note that $|L| = r p$.)
\begin{align*}
(U_1)_{\ell, *} 
= & ~  ( S_2 \diag( W_{*,\ell} ) \wh{U}_1 )^\dagger \cdot S_2 \diag(W_{*,\ell})  A_{*,\ell} \\
= & ~ (  ( S_2 \diag( W_{*,\ell} ) \wh{U}_1 ) ( S_2 \diag( W_{*,\ell} ) \wh{U}_1)^\top )^{-1} \\
& ~ \cdot S_2 \diag( W_{*,\ell} ) \wh{U}_1 \cdot S_2 \diag(W_{*,\ell})  A_{*,\ell}.
\end{align*}

Recall the number of different $A_{*,\ell} \diag(W_{*,\ell})$ is at most $rp$, and the number of different $ \diag(W_{*,\ell})$ is at most $r$.

For each $s_2 \times k$ matrix $S_2 \diag( W_{*,i} ) \wh{U}_1$, we create $s_2 \times k$ variables to represent it. Thus, we create $r$ matrices,
\begin{align*}
\{ S_2 \diag( W_{*,i} ) \wh{U}_1 \}_{i \in R}.
\end{align*}
For simplicity, let $P_{2,i} \in \R^{k \times s_2}$ denote $S_2 \diag( W_{*,i} ) \wh{U}_1$. Then we can rewrite $\wh{U}_2$ as follows
\begin{align*}
    (\wh{U}_2^\top)_{*,i} = (P_{2,i} P_{2,i}^\top)^{-1}  P_{2,i} S_2 \diag(W_{*,i})  A_{*,i}.
\end{align*}

In a similar way, we can write $\wh{U}_2$. Overall, by creating $l = O(rk^2/\epsilon)$ variables, we have rational polynomials $\wh{U}_1(x)$ and $\wh{U}_2(x)$. Note that $\wh{U}_1(x)$ only has $rp$ different rows, and same for $\wh{U}_2(x)$.

Putting it all together, we can write the objective function,
\begin{align*}
\min_{x \in \R^l} & ~ \| ( \wh{U}_1(x) \wh{U}_2(x)^\top - A ) \circ W \|_F^2 \\
\mathrm{~s.t.~}& ~ h_{1,i}(x) \neq 0 ,  \forall i \in [r] \\
& ~ h_{2,i}(x) \neq 0, \forall i \in [r].
\end{align*}

Note that $\wh{U}_1(x) \wh{U}_{2}(x) \circ W$ only has $rp$ distinct rows. Also, $A \circ W$ only has $rp$ distinct rows. Writing down the objective function $\| ( \wh{U}_1(x) \wh{U}_2(x)^\top - A ) \circ W \|_F^2$ only requires $n (rp) \cdot \poly(kr/\epsilon)$ time.

Combining the binary search explained in Lemma~\ref{lem:warmup} 
with the lower bound on cost (Theorem~\ref{thm:lower_bound_on_cost}) we obtained, we can find the solution for the original problem in time,
\begin{align*}
(np \cdot \poly(kr/\epsilon) + n 2^{\wt{O} (rk^2 / \epsilon) }  ) \cdot n^{\gamma} = n^{1+\gamma} p \cdot 2^{\wt{O}(rk^2/\epsilon)}.
\end{align*}

\end{proof}
\section{RECOVER A SOLUTION}\label{sec:recover_solution}

We state our results and proof of recovering a solution.
\begin{theorem}\label{thm:recover_solution}
Given a matrix $A$ and $W$, each entry can be written using $O(n^{\gamma})$ bits for $\gamma >0$.
Given $A \in \R^{n \times n}, W \in \R^{n \times n}, 1 \leq k \leq n$ and $ \epsilon \in (0,0.1)$. 
Assume $W$ has $r$ distinct columns and rows. 
Assume $A \circ W$ has at most $r \cdot p$ distinct columns and at most $r \cdot p$ distinct rows. 
Let $\OPT$ be defined as Definition~\ref{def:opt}. 
Then, with high probability,  
one can output two matrices $U,V \in \R^{n \times k}$ in time $n^{1+\gamma} \cdot 2^{O (k^2 r / \epsilon )} $ such that 
\begin{align*}
    \| (UV^\top- A) \circ W \|_F^2 \leq(1+\epsilon) \OPT .
\end{align*}

Further, if we choose (1) $k^2 r = O(\log n /\log\log n)$, (2) $\epsilon \in (0,0.1)$ to be a small constant, (3) $p = n^{o(1)}$ and (4) $\gamma = o(1)$, then the running time becomes $n^{1+o(1)}$.
\end{theorem}

\begin{proof}

Here, we show how to recover an approximate solution, not only the value of $\OPT$.

The idea is to recover the entries of $U$ and $V$ one by one and use the algorithm from the previous section for the corresponding decision problem. We initialize the semialgebraic set to be
\begin{align*}
    S= \{x \in \R^l \mid q(x) \neq 0, p(x) \leq \Lambda q(x) \}.
\end{align*}
We start by recovering the first entry of $U$. We perform the binary search to localize the entry, which takes $ \log  ( 2^{n^{\gamma}} )$ invocations of the decision algorithm. For each step of binary search, we use Theorem~\ref{thm:jpt13} to determine whether the following semi-algebraic set $S$ is empty or not,
\begin{align*}
S \cap \{U_{1,1}(x) \geq \wh{U}_{1,1}^{-}, U_{1,1}(x) \leq \wh{U}_{1,1}^{+} \}.
\end{align*}
After that, we declare the first entry of $U$ to be any point in this interval.  Then, we add an equality constraint that fixes the entry of $\wh{U}$ to this value and add a new constraint into $S$ permanently, e.g., $S  \leftarrow S \cap \{U_{1,1}(x)=\wh{U}_{1,1} \}$. Next, we repeat the same with the second entry of $U$ and so on.

This allows us to recover a solution of cost at most $(1+ \epsilon) \OPT$ in time
\begin{align*}
n^{1+\gamma} \cdot p \cdot 2^{O (k^{2} r / \epsilon )} .
\end{align*}
If we choose $\gamma = o(1)$, $\epsilon = \Theta(1)$, $p =n^{o(1)}$ and $k^2 r = O(\log n / \log\log n )$, then the running time becomes $n^{1+o(1)}$.
Thus, we complete the proof.
\end{proof}

\section{CONCLUSION}\label{sec:conclusion}

We showed that the weighted low-rank approximation problem could be solved in almost linear $n^{1+o(1)}$ time when the weighted matrix $W$ has few distinct columns and rows and $W \circ A$ has few distinct columns and rows. This demonstrates that truly subquadratic time is achievable for a dense regime not previously known to be tractable. 
Future work could generalize the assumptions and explore applications of the algorithm.

\section{LIMITATIONS}\label{sec:limitation}
Although our algorithm can achieve a good approximation ratio in theory, it may not achieve the ideal result in practical implementation due to various factors such as computational resource limitation, data quality problems, or incomplete model specification.

\section{SOCIETAL IMPACT}\label{sec:impact}
In this paper, we introduce an algorithm that can solve weighted low-rank approximation problems in near-linear time under certain conditions. Our paper is purely theoretical and empirical in nature (a mathematics problem), and thus, we foresee no immediate negative ethical impact. 
Our algorithms can improve the efficiency of data processing and model training, allowing complex algorithms to be run in resource-limited environments and accelerating scientific research and technological innovation. By reducing computing requirements, energy consumption can be reduced, and the environment can benefit.

\ifdefined\isarxiv
\else
\bibliography{ref}
\bibliographystyle{plainnat}
\input{20_checklist}
\fi



\ifdefined\isarxiv
\bibliographystyle{alpha}
\bibliography{ref}

\else

\fi




\end{document}